\documentclass[sigconf]{acmart}

\usepackage{subcaption}
\usepackage{algorithm}
\usepackage{algorithmicx}
\usepackage{algpseudocode}
\usepackage{multirow}
\usepackage{arydshln}
\usepackage{makecell}
\usepackage{xcolor}
\newtheorem{theorem}{Theorem}
\newcommand{\pr}{\textrm{Pr}}

\AtBeginDocument{%
  }

\copyrightyear{2026}
\acmYear{2026}
\setcopyright{cc}
\setcctype{by}
\acmConference[KDD '26]{Proceedings of the 32nd ACM SIGKDD Conference on Knowledge Discovery and Data Mining V.2}{August 09--13, 2026}{Jeju Island, Republic of Korea}
\acmBooktitle{Proceedings of the 32nd ACM SIGKDD Conference on Knowledge Discovery and Data Mining V.2 (KDD '26), August 09--13, 2026, Jeju Island, Republic of Korea}
\acmDOI{10.1145/3770855.3817759}
\acmISBN{979-8-4007-2259-2/2026/08}

\begin{document}

\title{Scalable Counterfactual Risk Estimation for Rare Events in Longitudinal Data}

\author{Xiaohui Yin}
\authornote{These authors contributed equally to this work.}
\affiliation{%
  \institution{University of Connecticut}
  \city{Storrs}
  \state{Connecticut}
  \country{USA}}

\author{Avijit Mitra}
\authornotemark[1]
\affiliation{%
  \institution{University of Massachusetts Amherst}
  \city{Amherst}
  \state{Massachusetts}
  \country{USA}}

\author{Ying Zhou}
\authornote{Corresponding author.}
\affiliation{%
  \institution{University of Connecticut}
  \city{Storrs}
  \state{Connecticut}
  \country{USA}}
\email{yzhou@uconn.edu}

\author{Kun Chen}
\affiliation{%
  \institution{University of Connecticut}
  \city{Storrs}
  \state{Connecticut}
  \country{USA}}

\author{Hong Yu}
\affiliation{%
  \institution{University of Massachusetts Lowell}
  \city{Lowell}
  \state{Massachusetts}
  \country{USA}}

\renewcommand{\shortauthors}{Xiaohui Yin, Avijit Mitra, Ying Zhou, Kun Chen, \& Hong Yu}

\begin{abstract}
  Estimating the causal effect of time-varying treatments on survival outcomes in large observational studies is computationally demanding, particularly when outcomes are rare. While g-formula–based methods such as the iterative conditional expectation (ICE) estimator provide a principled framework for longitudinal causal inference, they become computationally expensive, especially when bootstrap-based variance estimation is required. In addition, outcome rarity at each time point induces severe class imbalance, leading to instability and convergence issues in logistic regression and related models. To address these challenges, we propose a principled subsampling and reweighting strategy for longitudinal survival data that can be applied to a range of existing causal effect estimators in this setting, including the ICE estimator. The proposed method substantially reduces computational burden while preserving consistency and improving estimation stability in rare-outcome settings. We evaluate the method through simulations and validate it using a large-scale EHR cohort study on social and behavioral determinants of health (SBDH) and suicide risk, demonstrating its effectiveness for modeling rare outcomes in longitudinal data.
\end{abstract}



\maketitle

\section{Introduction}
\label{sec:intro}

Many observational studies in causal inference aim to estimate the causal effect of a time-varying treatment on a survival outcome. Examples include evaluating the impact of a long-term medication regimen on the survival of patients with chronic diseases \citep{hernan2006estimating}, assessing the effect of smoking cessation over time on lung cancer incidence \citep{kenfield2008comparison}, and studying the influence of lifestyle interventions on cardiovascular event-free survival \citep{knowler2002reduction}.

Under the standard causal assumptions of consistency, positivity, and exchangeability given measured confounders, the counterfactual probabilities of the time-to-event outcomes can be identified using the g-formula \citep{robins1986new}. Two widely-used forms of the g-formula are: (1) an expectation weighted by the joint distribution of covariates, treatments, and outcomes, and (2) an iterative conditional expectation over time. A straightforward approach to constructing estimators based on the g-formula is to first estimate each component and then plug these estimates into the g-formula. The estimator based on the first form of the g-formula is known as the \textit{Non-Iterative Conditional Expectation (NICE)} estimator, which typically requires modeling the joint distribution of the confounders, treatments, and outcomes over time. The estimator based on the second representation is the \textit{Iterative Conditional Expectation (ICE)} estimator, which relies on a sequence of conditional expectation models at each time point but avoids specifying the full joint distribution of confounders.

A major challenge in applying both the NICE and ICE estimators in large observational studies is the substantial computational burden, particularly when the outcome of interest is rare \citep{king2001logistic}. For example, in our analysis of a large-scale cohort study examining the effect of social and behavioral determinants of health (SBDH) on suicide with around 130,000 individuals, fitting logistic regression models at each time point is computationally intensive. This difficulty arises primarily from the rarity of suicide and the large sample size, both of which contribute to slow convergence when estimating models for binary outcomes. In addition, because closed-form expressions for standard errors are difficult to derive in this setting, inference typically relies on bootstrap procedures, further increasing the computational burden through repeated resampling and model fitting. These challenges are especially pronounced in longitudinal settings, where models must be estimated repeatedly across time. In some applied environments, such as secure data platforms with limited storage or memory, repeatedly operating on full longitudinal cohorts may also be infeasible, further motivating the need for scalable alternatives. 

To address these challenges, we adopt a subsampling perspective for rare-event modeling, a well-established and cost-effective strategy for reducing computational burden while preserving statistical efficiency \citep{breslow1996statistics,van2008estimation,Wang2020}. In particular, case–control designs can be viewed as a principled form of outcome-dependent subsampling that is especially effective when events are rare. However, existing case–control approaches are typically defined locally at individual event times or risk sets, and are not designed to construct a coherent sampled dataset that respects the repeated, time-indexed structure of longitudinal data with time-varying treatments and covariates. In such settings, sampling decisions at one time point impose structural constraints on subsequent time points through survival, temporal ordering, and within-subject dependence.

We therefore propose longitudinal case–control sampling and reweighting strategies that explicitly account for these cross-time dependencies. Our approach constructs time-specific samples while respecting risk-set structure and longitudinal constraints, enabling computationally efficient learning under rare outcomes without compromising validity. While we illustrate the method through its integration with g-formula–based estimators for time-varying treatment effects, the proposed sampling framework is more general and can be applied to a broad class of longitudinal causal and predictive methods. We demonstrate the effectiveness of the proposed approach through extensive simulation studies and a real-world application involving suicide as a rare outcome, showing substantial reductions in computational cost while maintaining estimation accuracy and stability.

\section{Related Work} \label{sec:lit}
A range of methodological frameworks have been developed to estimate causal effects in longitudinal settings where treatments, covariates, and outcomes evolve over time. These include the g-formula, marginal structural models (MSMs), and structural nested models (SNMs), which are designed to address challenges such as time-varying confounding and to evaluate causal effects under history-dependent treatment strategies \citep{HernanRobins2020,Wen2021,Robins2000,Vansteelandt2014}. All three approaches aim to correct for time-dependent confounding that arises when past exposures influence subsequent covariates, but they differ in how they achieve this goal.
MSMs rely on inverse-probability weighting to construct a pseudo-population in which treatment assignment is independent of past confounders. These models are relatively straightforward to implement using standard regression tools and yield marginal causal effects. However, their performance can be compromised by highly variable or extreme weights, especially in settings with limited overlap, leading to unstable or biased estimates.
SNMs estimate causal effects via g-estimation of blip functions, which quantify the change in the counterfactual outcome induced by deviating from a reference treatment path at each time point, conditional on the observed history. With correctly specified treatment and outcome models, g-estimation attains the semiparametric efficiency bound and yields history-specific causal contrasts. In practice, however, SNMs rely on computationally demanding, sometimes unstable iterative equations, are highly sensitive to misspecification of either model, and yield blip parameters that have limited direct clinical interpretability.

The parametric g-formula models each component of the data-generating process and can be implemented through two algebraically equivalent forms which will be discussed in Section \ref{sec:methodology}. Although it requires correct specification of all sub-models, the g-formula avoids the weight instability inherent to MSMs and the iterative g-estimation burden of SNMs, accommodates both continuous and time-to-event outcomes, and provides clear population-level contrasts for complex static or dynamic treatment strategies. Because of this finite-sample stability, modelling flexibility, and transparent counterfactual interpretation, we use the g-formula to examine how evolving SBDH influence suicide-related outcomes over time.

In addition to these well-established frameworks, subsampling has emerged as an important strategy for improving the computational efficiency of estimation in large-scale data. Rather than analyzing the entire cohort, one can draw informative subsets of the data to substantially reduce computational burden while still enabling valid inference. In rare-event settings, theoretical and empirical results show that retaining all cases while subsampling a sufficiently large number of controls, together with appropriate reweighting, yields estimators that are consistent and, under suitable asymptotic regimes, can achieve the same asymptotic variance as the corresponding full-sample maximum likelihood estimator \citep{Wang2020}. This indicates that, when outcomes are highly imbalanced, most of the statistical information is concentrated in the rare cases, and additional controls contribute diminishing marginal information beyond a certain point. Consequently, subsampling provides a principled, efficiency-preserving approach for scaling g-formula–based methods and related longitudinal estimators to massive datasets where full-sample estimation is computationally prohibitive.

\section{Methodology}\label{sec:methodology}

\paragraph{Study Design}

We consider a longitudinal study with discrete time points \( j = 0, \ldots, T \), where individuals are followed over time. At each time \( j \), covariates \( L_j \) and treatment \( A_j \) are observed, followed by the binary censoring indicator \( C_{j+1} \) and binary event indicator \( Y_{j+1} \). Censoring (\( C_{j+1} = 1 \)) indicates that the individual is no longer observed from time \( j+1 \) onward. This may occur in practice due to loss to follow-up, withdrawal, or administrative censoring at the end of the study period. We define $Y_j$ as the event-at-time indicator: $Y_j=1$ if the event occurs exactly at time $j$. Thus, if an individual experiences the event at time \( k \), then \( Y_k = 1 \) and \( Y_j = 0 \) for all $j < k$. The data are temporally ordered as \( (L_j, A_j, C_{j+1}, Y_{j+1}) \), with \( Y_0 = C_0 = 0 \) and \( \bar{L}_{-1} = \bar{A}_{-1} = \emptyset \) by convention.
As an illustration, an individual who experiences the event at time 2 would have the observed data sequence
\(
(L_{-1}= \emptyset, A_{-1}= \emptyset, C_0=0, Y_0=0, L_0=l_0, A_0=a_0, C_1=0, Y_1=0, L_1=l_1, A_1=a_1, C_2=0, Y_2=1).
\)

We use \( Y_t^g \) to denote the potential outcome at time \( t \) had an individual followed a deterministic treatment strategy \( g \).
A deterministic strategy \( g \) specifies the treatment \( A_j \) to be assigned at each time \( j \) based on the observed history up to that point, \( (\bar{L}_j, \bar{A}_{j-1}) \). 
Our goal is to identify and estimate $E(Y_t^g)$, the mortality risk by time \( t \) if all individuals in the population were to follow strategy \( g \).

\paragraph{G-formula} 
Identification of $E(Y_t^g)$ using the g-formula \citep{robins1986new} relies on three key assumptions: \\
1. \textit{Consistency:} If an individual’s observed treatment history matches the strategy \( g \) up to time \( j \), i.e., \( \bar{A}_j = \bar{A}_j^g \), then their observed covariates and outcomes also equal the counterfactuals: \( \bar{L}_j = \bar{L}_j^g \) and \( \bar{Y}_{j+1} = \bar{Y}_{j+1}^g \).    \\
2. \textit{Exchangeability:} At each time \( j \), future counterfactual outcomes are independent of treatment and censoring decisions, conditional on past covariate and treatment history and being uncensored and event-free: 
    \[
    (Y_{j+1}^g, \ldots, Y_T^g) \perp (A_j, C_{j+1}) \mid \bar{L}_j = \bar{l}_j, \bar{A}_{j-1} = \bar{a}_{j-1}^g, C_j = Y_j = 0.
    \]
3. \textit{Positivity:} If a covariate and treatment history occurs with positive probability under the observed data, then the probability of receiving the treatment assigned by strategy \( g \) and remaining uncensored at the next time must also be positive:
\begin{align*}
    & f_{\bar{L}_j, \bar{A}_{j-1}, C_j, Y_j}(\bar{l}_j, \bar{a}_{j-1}^g, 0, 0) > 0 \\
    \Rightarrow\;& 
    f_{A_j, C_{j+1} \mid \bar{L}_j, \bar{A}_{j-1}, C_j, Y_j}(a_j^g, 0 \mid \bar{l}_j, \bar{a}_{j-1}^g, 0, 0) > 0.
\end{align*}

Under these assumptions, the counterfactual risk \( E(Y_T^g) \) can be identified via the g-formula:
\begin{align}
&E(Y_T^g) = \sum_{\bar{l}_{T-1}} \sum_{t=1}^T P(Y_t=1 \mid Y_{t-1}=C_t=0, \bar{L}_{t-1}=\bar{l}_{t-1}, \bar{A}_{t-1}=\bar{a}_{t-1}^g)  \nonumber  \\
&\prod_{s=0}^{t-1} P(Y_s=0 \mid  Y_{s-1}  =C_s=0, \bar{L}_{s-1}=\bar{l}_{s-1}, \bar{A}_{s-1}=\bar{a}_{s-1}^g) \label{eq:NICE} \\
&\qquad f(l_s \mid Y_s=C_s=0, \bar{L}_{s-1}=\bar{l}_{s-1}, \bar{A}_{s-1}=\bar{a}_{s-1}^g). \nonumber
\end{align}
Equation (\ref{eq:NICE}) is known as the non-iterative conditional expectations (NICE) representation. 
Alternatively, using iterated conditional expectations (ICE), \( E(Y_T^g) \) can be identified as:
\begin{align}\label{eq:ICE}
\begin{split}
E(Y_T^g) = E_{f_{L_0}}\Bigl(E_{f_{Y_1}}\Bigl[Y_1 + \cdots E_{f_{Y_{T-1}}}\bigl\{ 
Y_{T-1}(1-Y_{T-2}) 
+ E_{f_{L_{T-1}}}\big[ \\ E_{f_{Y_T}}\{ Y_T(1-Y_{T-1}) \mid 
\bar{Y}_{T-1},C_T=0, 
\bar{L}_{T-1}, \bar{A}_{T-1}=\bar{A}^g_{T-1}\}\mid \\
\bar{Y}_{T-1}, C_{T-1}=0,\bar{L}_{T-2}, \bar{A}_{T-2}=\bar{A}^g_{T-2} \big] \mid 
\bar{Y}_{T-2},C_{T-1}=0, \bar{L}_{T-2}, \\\bar{A}_{T-2}=\bar{A}^g_{T-2}\bigr\}
...\mid C_1 = 0, A_0=A_0^g,L_0\Bigr]  \Bigr).
\end{split}
\end{align}

In this paper, we focus on the ICE formula (\ref{eq:ICE}). However, our methods are also applicable to the NICE estimator, as demonstrated in our simulation studies.  Based on (\ref{eq:ICE}),  an algorithm for estimating $E(Y^g_T)$ is presented in Algorithm \ref{alg:ice_full}. If the parametric models in lines 1 and 6 in Algorithm~\ref{alg:ice_full} are correctly specified, this estimator will be consistent.

\begin{algorithm}[htbp]
\caption{ICE Estimator for Estimating $E(Y_T^g)$}
\label{alg:ice_full}
\begin{algorithmic}[1]
\State Regress $Y_T$ on $\bar{L}_{T-1}$ and $\bar{A}_{T-1}$ among individuals with $Y_{T-1} = C_T = 0$, and estimate parameter $\theta_{T,T-1}$.
\State Obtain predicted values $\hat{h}^g_{T,T-1}$ from 
\(
E(Y_T \mid Y_{T-1}=C_T=0, \bar{L}_{T-1}, \bar{A}_{T-1}=\bar{A}^g_{T-1}; \hat{\theta}_{T,T-1})
\)
by fixing $\bar{A}_{T-1} = \bar{A}^g_{T-1}$ for all individuals with $Y_{T-1}=C_{T-1}=0$. Set $q=2$.
\While{$q \leq T$}
    \State Let $k = T - q$.
    \State Define
    \[
    \hat{Q}^g_{T,k+1} = 
    \begin{cases}
    \hat{h}^g_{T,k+1} & \text{if } Y_{k+1}=0,\\
    1 & \text{if } Y_{k+1}=1.
    \end{cases}
    \]
    \State Regress $\hat{Q}^g_{T,k+1}$ on $\bar{L}_k$ and $\bar{A}_k$ among individuals with $Y_k = C_{k+1} = 0$, and estimate parameter $\theta_{T,k}$.
    \State Obtain predicted values $\hat{h}^g_{T,k}$ from
    \(
    E(\hat{Q}^g_{T,k+1} \mid Y_k=C_{k+1}=0, \bar{L}_k, \bar{A}_k=\bar{A}^g_k; \hat{\theta}_{T,k})
    \)
    by fixing $\bar{A}_k = \bar{A}^g_k$ among individuals with $Y_k=C_k=0$.
    \State Set $q=q+1$.
\EndWhile
\State Average $\hat{h}^g_{T,0}$ over all individuals to estimate $E(Y_T^g)$.
\end{algorithmic}
\end{algorithm}

\paragraph{An Illustrative Example} 

In Algorithm~\ref{alg:ice_full}, models are estimated at each time point using all available samples. In large cohorts with rare outcomes, the resulting excess of controls leads to substantial computational cost, particularly when models must be fit repeatedly across multiple time points and within bootstrap procedures.
To improve computational efficiency, we propose to apply Algorithm~\ref{alg:ice_full} to a longitudinal case--control subsample that retains all cases while subsampling controls at each time point, reducing computational burden without sacrificing estimation reliability.

To illustrate the key ideas, we consider a simplified setting with two time points, $t=0,1$, in the absence of censoring. Suppose the total number of individuals is $N$, among whom \( C \) experience the event. Let $c_t$ denote the number of individuals who experience the event at time $t$, so that $C = c_0+c_1$. At each time point $t$, we include all cases occurring at $t$ and randomly sample $Jc_t$ controls with replacement from individuals who remain event-free through time $t$, where $J$ denotes a fixed sampling ratio. The resulting case--control subsample therefore contains $(J+1)(c_0+c_1)$ observations. If an individual is sampled at both time points, each appearance is treated as a distinct observation.

Suppose our goal is to estimate the expectation of a covariate $X$. Using the full sample, a natural estimator is $\hat{E}(X)=N^{-1}\sum_{i=1}^N X_i$. 
We now show how a longitudinal case–control subsample with appropriate reweighting can be constructed to recover this quantity in expectation.
The original sample is partitioned into three mutually exclusive groups: 
\( G_0 = \{ i : Y_0 = 1 \} \) (those who die at \( t=0 \)); 
\( G_1 = \{ i : Y_1=1, Y_0=0 \} \) (those who die at \( t=1 \)); and 
\( G_2 = \{ i : Y_1=0 \} \) (those alive at \( t=1 \)).
Each individual belongs to exactly one group. 
Table \ref{tab:t1} summarizes the structure of the case–control subsample. 
The column ``weight'' gives the weight associated with each category, the column ``group'' lists the groups included in the category, and the column ``count'' indicates the number of times a person in this category is counted. Using time 0 control as an example, the probability that an individual who has not experienced the event at time 0 is selected into the subsample is $(Jc_0)/(N-c_0)$. Therefore, the count is the product of the weight and this selection probability. 
\begin{table}[htbp]
  \caption{Subsample breakdown with two time points and no censoring}
  \vspace{2mm}
  \label{tab:t1}
    \centering
    \begin{tabular}{lcccc}\toprule
        & size & weight & group & count \\\midrule
        time 0 case & $c_0$ & $w_0$ & $G_0$ & $w_0\cdot 1$\\
        time 0 control & $Jc_0$ & $m_0$ & $G_1,\, G_2$ & $m_0\cdot \frac{Jc_0}{N-c_0}$\\
        time 1 case & $c_1$ & $w_1$ & $G_1$ & $w_1\cdot 1$ \\
        time 1 control & $Jc_1$ & $m_1$ & $G_2$ & $m_1\cdot \frac{Jc_1}{N-c_0-c_1}$ \\\bottomrule
    \end{tabular}
\end{table}

Because individuals may be sampled at multiple time points, unbiased estimation requires balancing their expected total weighted contribution across groups.
An individual in \(G_0\) is always included and thus contributes \(w_0 \cdot 1\). An individual in \(G_1\) may appear as a time 0 control with probability \(Jc_0/(N-c_0)\) and as a time 1 case with probability 1, so its expected contribution is \(w_1 \cdot 1 + m_0 \cdot Jc_0/(N-c_0)\). An individual in \(G_2\) may appear as a time 0 control with probability \(Jc_0/(N-c_0)\) and as a time 1 control with probability \(Jc_1/(N-c_0-c_1)\), giving an expected contribution of \(m_1 \cdot Jc_1/(N-c_0-c_1) + m_0 \cdot Jc_0/(N-c_0)\). In the full sample, each individual is counted exactly once. To ensure that the case--control subsample is representative of the full sample, we require that the expected contribution per individual be equal across groups. This balancing condition leads to
\[
w_0 = w_1 + m_0 \frac{Jc_0}{N-c_0} = m_1\frac{Jc_1}{N-c_0-c_1}+m_0 \frac{Jc_0}{N-c_0},
\]
which can be rewritten as 
\begin{align*}
w_1 = m_1 \frac{Jc_1}{N-c_0-c_1}, \quad
w_0-w_1 = m_0\frac{Jc_0}{N-c_0}.
\end{align*}
Since there are 2 equations with 4 parameters $(w_0,w_1,m_0,m_1)$, we have 2 degrees of freedom.
One convenient normalization is to set $w_0=1/N$, which ensures that the total weight across the subsample equals one: \(c_0w_0 + c_1w_1 + Jc_1m_1 + Jc_0m_0 = 1\).
If we further let $w_1 = 1/N$, then $m_0=0$ and $m_1 = (N-c_0-c_1)/(JNc_1)$. It is also possible to set $w_1$ to other values and obtain corresponding values of $m_0$ and $m_1$. Additional examples of weight combinations are provided in the Appendix \ref{appd:additional_weights}.

When estimating a statistical model, such as a logistic regression, parameter estimation is typically carried out by solving an estimating equation. In many cases, the estimating equation is from the maximum likelihood estimation process by setting the gradient of the log-likelihood to zero. 
We now show that the weighting strategy introduced above naturally extends to such estimating equations. We denote $O^*$ as the whole population, and $O=(O_{0c},O_{0m},O_{1c},O_{1m})$ as the longitudinal case--control sample, where subscripts indicate time point $(0,1)$ and case/control status $(c,m)$. Suppose the target parameter is defined by the population-level estimating equation $E^*_0\{S_0(O^*)\}=0$, where $E^*_0$ denotes expectation with respect to the true distribution of the whole population, and $S_0$ is the original estimating equation. We define the weighted estimating equation based on the case–control sample as 
$$S(O) = c_1w_1S_0(O_{1c})+Jc_1m_1 S_0(O_{1m})+c_0w_0 S_0(O_{0c})+ Jc_0m_0 S_0(O_{0m}).$$
Then, under the proposed sampling and weighting scheme, it holds that
\begin{align}\label{eq:estimating_equation}
E\{S(O)\}=w_0 N E_0^*\{S_0(O^*)\}=0,
\end{align}
where the expectation $E$ is taken with respect to the sampling distribution induced by the longitudinal case--control design.
The proof of (\ref{eq:estimating_equation}) is relegated to the Appendix \ref{appd:proof_eq}. 
Equation (\ref{eq:estimating_equation}) implies that, with appropriate weighting, the population estimating equation \( E_0^*\{S_0(O^*)\} = 0 \) if and only if the weighted estimating equation satisfies \( E\{S(O)\} = 0 \). Consequently, the target model parameters can be consistently estimated by solving the empirical estimating equation \( \hat{E}\{S(O)\} = 0 \) using the case--control subsample.

\paragraph{General Results with Censoring} Now we extend the setting to $k$ time points, $t = 0, 1, ..., k-1$, and allow for censoring. Let $c_t$ denote the number of individuals who experience the event at time $t$, and let $s_t$ denote the number of individuals who are
censored at time $t$. We allow the sampling ratios for controls and censored individuals, denoted by $J_t$ and $K_t$ respectively, to vary across time points. 
At each time point $t$, we include all cases occurring at $t$. In addition, we sample $K_tc_t$ from those who are censored at time $t$ except for the last time point, and sample $J_tc_t$ controls from individuals who have not experienced the event prior to or at time $t$. Censored individuals are not sampled at the last time point $t=k-1$, because their outcome \( Y_{k-1} \) is unobserved, which prevents their use in fitting logistic or other classification models in the first step of Algorithm \ref{alg:ice_full}. At earlier time points, however, censored individuals can still contribute to estimating counterfactual risks and may therefore be included. Under this design, the case--control subsample has a size of $\sum^{k-1}_{t=0} c_t+ \sum^{k-1}_{t=0}J_tc_t + \sum^{k-2}_{t=0}K_tc_t$.

The population can be partitioned into $2k+1$ mutually exclusive groups: $G_0, G_1, \ldots, G_{2k-1},\allowbreak G_{2k}$, where $G_{2i}$ is the group of people who are censored at time $i$ for $i=0,\ldots,k-1$, $G_{2i+1}$ is the group of people who experience the event at time $i$ for $i=0,\ldots,k-1$, and $G_{2k}$ is the group of people who remain event-free and uncensored through time $k-1$. A detailed description of the resulting subsample for $k$ time points is provided in Table~\ref{tab:t2}, with the ``count'' column constructed in the same manner as in Table \ref{tab:t1}. We use $w_t$ to denote the weights for cases at time $t$, $m_t$ for sampled controls, and $\ell_t$ for sampled censored individuals.

\begin{table*}[!t]
  \caption{Subsample breakdown with $k$ time points and censoring}
  \vspace{2mm}
  \label{tab:t2}
  \centering
  \begin{tabular}{lcccc}\toprule
    & size & weight & group & count \\\midrule
    time 0 censored & $K_0c_0$ & $\ell_0$ & $G_0$ & $\ell_0\cdot \frac{K_0c_0}{s_0}$\\
    time 0 case & $c_0$ & $w_0$ & $G_1$ & $w_0\cdot 1$\\
    time 0 control & $J_0c_0$ & $m_0$ & $G_2,\ldots,G_{2k}$ & $m_0\cdot \frac{J_0c_0}{N-c_0-s_0}$\\
    time 1 censored & $K_1c_1$ & $\ell_1$ & $G_2$ & $\ell_1\cdot \frac{K_1c_1}{s_1}$\\
    time 1 case & $c_1$ & $w_1$ & $G_3$ & $w_1\cdot 1$ \\
    time 1 control & $J_1c_1$ & $m_1$ & $G_4,\ldots,G_{2k}$ & $m_1\cdot \frac{J_1c_1}{N-c_0-s_0-c_1-s_1}$ \\
    $\qquad\vdots$ & $\vdots$ & $\vdots$ & $\vdots$ & $\vdots$\\
    time $k\!-\!1$ censored & $K_{k-1}c_{k-1}$ & $\ell_{k-1}$ & $G_{2k-2}$ & $\ell_{k-1}\cdot \frac{K_{k-1}c_{k-1}}{s_{k-1}}$\\
    time $k\!-\!1$ case & $c_{k-1}$ & $w_{k-1}$ & $G_{2k-1}$ & $w_{k-1}\cdot 1$ \\
    time $k\!-\!1$ control & $J_{k-1}c_{k-1}$ & $m_{k-1}$ & $G_{2k}$ & $m_{k-1}\cdot \frac{J_{k-1}c_{k-1}}{N-\sum_{j=0}^{k-1}c_j-\sum_{j=0}^{k-1}s_j}$ \\\bottomrule
  \end{tabular}
\end{table*}

Requiring that each individual in groups $G_0$ through $G_{2k}$ have the same expected total weighted contribution leads to the following system of equations:
\begin{align}\label{eq:weights}
\begin{split}
w_t &= \ell_t \frac{K_tc_t}{s_t}, \quad t = 0,\ldots,k-1, \\
w_{k-1} &= m_{k-1} \frac{J_{k-1}c_{k-1}}{N-\sum_{j=0}^{k-1}c_j - \sum_{j=0}^{k-1}s_j}, \\
w_t-w_{t+1} &= m_t\frac{J_tc_t}{N-\sum_{j=0}^{t}c_j - \sum_{j=0}^{t}s_j}, \quad t = 0,\ldots,k-2.
\end{split}
\end{align}
Since there are $2k$ equations with $3k$ parameters $(w_0,\ldots,w_{k-1},\allowbreak m_0,\ldots,m_{k-1},\ell_0,\ldots,\ell_{k-1})$, we have $k$ degrees of freedom.
If we let $w_0=1/N$, then the total weight across the subsample equals one:
\begin{align*}
\sum_{t=0}^{k-1}c_tw_t + \sum_{t=0}^{k-1} J_tc_tm_t  + \sum_{t=0}^{k-1}K_tc_t\ell_t = w_0N = 1.
\end{align*}
Letting \( K_t c_t = s_t \) for \( t = 0, \ldots, k-1 \), we have \( \ell_t = w_t \). This implies that, at each time point \( t \), if all censored individuals are included, they should be assigned the same weights as the cases at time \( t \). We adopt this design in our paper. Under this choice, the system reduces to \( k \) equations in \( 2k \) unknowns, yielding infinitely many valid weighting schemes. In this paper, we focus on one particular choice:
$$ \ell_0 = w_0 = \frac{c_0}{N-s_0}, \, m_0 = \frac{1}{J}\frac{N-c_0-s_0}{N-s_0}, \, \ell_t = w_t = m_t = 0, $$ for $t=1,\ldots,k-1$. It is straightforward to verify that this set of weights satisfies (\ref{eq:weights}). This choice corresponds to a simple weighting scheme that includes only individuals from time \( t = 0 \).
Here, $t=0$ is a local index used in the complete-data subproblem and should not be confused with the first time point of the original longitudinal study. 
When applying Algorithm~\ref{alg:ice_full} within the longitudinal case--control framework, each global time point at which a model is fitted is relabeled as the local baseline $t=0$ through a purely notational reindexing of time. 
In addition, because an individual who experiences the event at a later time point may be sampled as a control at earlier time points, the same individual may contribute multiple observations with different roles and receive different weights across time. The resulting ICE estimator based on the constructed case–control subsample is summarized in Algorithm~\ref{alg:case_control}.

To connect the local weighting derivation with Algorithm~\ref{alg:case_control}, note that Algorithm~\ref{alg:case_control} iterates backward over the global time points $k=T-1,\ldots,0$. At each step, the current global time point $k$ is relabeled as the local baseline time $t=0$. Under the simple choice $\ell_t=w_t=m_t=0$ for $t\geq 1$, only the cases, censored individuals, and sampled controls at the current global time point $k$ contribute to the estimating equation. The weights $w_k,\ell_k,m_k$ in Algorithm~\ref{alg:case_control} are therefore the local weights $w_0,\ell_0,m_0$ recomputed at each step.

\begin{algorithm}[ht]
\caption{Case--Control Sampling and Estimation Procedure with Censoring}
\label{alg:case_control}
\begin{algorithmic}[1]
\State Generate a case--control sample:
    \For{each time point $t$}
        \State Keep all $c_t$ cases at time $t$ (i.e., those who die at time $t$).
        \If{$t < k-1$}
            \State Randomly draw $J_tc_t$ controls from those alive at time $t$ with replacement.
            \State Assign weight $w_t$ to each case at time $t$.
            \State Assign weight $w_t$ to each censored individual at time $t$.
            \State Assign weight $m_t$ to controls sampled at time $t$.
        \ElsIf{$t = k-1$}
            \State Randomly draw $J_{k-1}c_{k-1}$ controls from those alive (excluding censored individuals) at time $k-1$ with replacement.
            \State Assign weight $w_{k-1}$ to each case at time $k-1$.
            \State Assign weight $m_{k-1}$ to controls sampled at time $k-1$.
        \EndIf
    \EndFor
\State Fit a weighted regression of $Y_T$ on $\bar{L}_{T-1}$ and $\bar{A}_{T-1}$ among individuals with $Y_{T-1}=C_T=0$ using the associated weights, and estimate parameter $\theta_{T,T-1}$.
\State Obtain predicted values $\hat{h}^g_{T,T-1}$ from
\[
E(Y_T \mid Y_{T-1}=C_T=0, \bar{L}_{T-1}, \bar{A}_{T-1}=\bar{A}^g_{T-1}; \hat{\theta}_{T,T-1})
\]
by fixing $\bar{A}_{T-1} = \bar{A}^g_{T-1}$ among individuals with $Y_{T-1}=C_{T-1}=0$. Set $q=2$.
\While{$q \leq T$}
    \State Let $k = T - q$.
    \State Define
    \[
    \hat{Q}^g_{T,k+1} =
    \begin{cases}
    \hat{h}^g_{T,k+1} & \text{if } Y_{k+1}=0, \\
    1 & \text{if } Y_{k+1}=1.
    \end{cases}
    \]
    \State Fit a weighted regression of $\hat{Q}^g_{T,k+1}$ on $\bar{L}_k$ and $\bar{A}_k$ among individuals with $Y_k=C_{k+1}=0$ using the associated weights, and estimate parameter $\theta_{T,k}$.
    \State Obtain predicted values $\hat{h}^g_{T,k}$ from
    \[
    E(\hat{Q}^g_{T,k+1} \mid Y_k=C_{k+1}=0, \bar{L}_k, \bar{A}_k=\bar{A}^g_k; \hat{\theta}_{T,k})
    \]
    by fixing $\bar{A}_k = \bar{A}^g_k$ among individuals with $Y_k = C_k = 0$.
    \State Set $q = q + 1$.
\EndWhile
\State Compute the weighted average of $\hat{h}^g_{T,0}$ using the generated weights to estimate $E(Y_T^g)$.
\end{algorithmic}
\end{algorithm}

\begin{theorem}\label{thm:1}
The estimator from Algorithm~\ref{alg:case_control} is consistent for \(E(Y_T^g)\) under the same assumptions that ensure the consistency of the ICE estimator in Algorithm~\ref{alg:ice_full}.
\end{theorem}
The proof of Theorem \ref{thm:1} is in Appendix \ref{appd:thm_proof}.

\paragraph{Practical Guidance on Weighting and Sampling Ratios}
Our framework admits a family of valid weighting schemes that yield consistent estimation as long as the weights satisfy the balance equations~(\ref{eq:weights}); however, different choices of weights and sampling ratios can influence efficiency. Because closed-form variance formulas for either the NICE or ICE g-formula implementation are prohibitively complex, analytic efficiency comparisons across weighting choices are difficult, so we follow the standard practice in applied g-formula work and obtain standard errors via nonparametric bootstrap resampling \citep{young2011comparative,keil2014parametric,Wen2021}. For the sampling ratios, we recommend a pragmatic procedure. One may begin with a moderate ratio, such as five controls per case, and then increase the time-specific ratio $J_t$ adaptively when events become sparse in later periods, so that each regression is fitted with a reasonable effective sample size. In applications, we recommend running a short pilot bootstrap over two or three candidate ratios to assess the marginal reduction in standard error relative to the added computational cost. This empirical tuning provides a practical way to balance statistical precision and computational efficiency, and in our numerical studies it preserves much of the precision of the full-data estimator while delivering substantial speed-ups.

\section{Numerical Examples}

\subsection{Simulated Longitudinal Data}\label{sec:sim}
We conducted simulation studies to compare the computational cost and estimation efficiency of the proposed subsampling estimator with its complete-data counterpart. The control-to-case sampling ratios were set at $1:5$, $1:10$, and $1:20$. 
Each dataset contained $N=30,000$ individuals followed over six time points. Baseline covariates \(L_b\) included four categorical variables, and at each time point \(t\), the data consisted of ten binary time-varying covariates \(L_t \in \mathbb{R}^{10}\), a binary treatment indicator \(A_t\), a binary censoring indicator \(C_t\), and a death indicator \(Y_t\). The observed data for each individual were $(L_b, L_0, A_0, C_1, Y_1, \dots, L_5, A_5, C_6, Y_6)$.
A first-order (lag-1) dependence structure was assumed for covariates, censoring, and outcomes, and all variables were generated at the individual level using logistic regression models, following the design of \citet{Wen2021}.

The details of data generation are as follows. Baseline covariates $ {L}_b$ were generated from multinomial distributions with the probability vectors $(1,2)$, $(1,1,1,1,1,1)$, $(5,2,1,1,4)$, and $(2,3,1,4)$, respectively; $ {L}_0\sim \mathrm{Ber}(0.5)$, and $A_0 \sim  \mathrm{Ber}\bigl(\mathrm{expit}( {L}_0^\top  {\eta})\bigr)$. Moreover, for $t\ge 1$,
\begin{align*}
C_t &\sim \mathrm{Ber}\bigl(\mathrm{expit}(-6 + A_{t-1} +  {\gamma}^\top {L}_{t-1})\bigr), \\
Y_t &\sim \mathrm{Ber}\bigl(\mathrm{expit}(-6 + 2 A_{t-1} +  {\beta}^\top {L}_{t-1})\bigr),\\
 {L}_t &\sim \mathrm{Ber}\bigl(\mathrm{expit}( {\eta}A_{t-1} +  {M} {L}_{t-1})\bigr), \\
A_t &\sim \mathrm{Ber}\bigl(\mathrm{expit}(A_{t-1} +  {\alpha}^\top  {L}_{t})\bigr),
\end{align*}
where the entries of the coefficient vectors $ {\alpha}$, $ {\beta}$, $ {\gamma}$, $ {\eta}$, and the matrix $ {M}$ were independently generated from a standard normal distribution and then fixed throughout the simulations. This setting yields a prevalence rate of about $1\%$.
We were interested in estimating the counterfactual risks under two deterministic treatment strategies: always treated and never treated, i.e., $\bar{A}_5=(1,1,1,1,1,1)$ and $\bar{A}_5=(0,0,0,0,0,0)$. 

We simulated 100 independent datasets under the same parameter configuration. For each dataset, we performed patient-level bootstrap resampling with \(B = 100\) replicates and applied both Complete and Case--control versions of the NICE and ICE estimators with logistic regression. 
We also explored more computationally intensive machine learning approaches with superlearner, an ensemble framework that combines multiple candidate algorithms to optimize predictive performance. 
For each estimator, we computed the bootstrap mean and standard deviation of the estimated target quantity. For the ICE method with logistic regression, we followed the R package \texttt{gfoRmula} \citep{mcgrath2020gformula} and implemented our case--control version. Superlearner models were implemented using the R package \texttt{SuperLearner} \citep{polley2019package} with base learners supporting weighted samples and probabilistic outcomes, including logistic regression and XGBoost.

\paragraph{Computation Time}
Computation times for logistic regression and superlearner were summarized in Table \ref{tab:sim:summary_time} and Table~\ref{tab:sim:summary_time_sl}. Each entry is the average of 10,000 analyses, that is, 100 bootstrap resamples applied to each of 100 independently simulated data sets. For the case--control version we report both (i) the total runtime, which includes construction of the subsample plus model fitting, and (ii) the model-fitting time alone, allowing the additional cost of the sampling step to be isolated. Each bootstrap resample and model fitting was conducted on a Linux HPC using a single node, equipped with either an Intel Skylake or AMD Epyc processor (Epyc64 or Epyc128 architecture). Each task was allocated 8 CPU cores and 4 GB of memory. The analyses were performed using R version 4.4.1, loaded via the module environment.

\begin{table}[htbp]
  \centering \scriptsize
  \caption{Simulation: Summary of the computation time (in secs) with logistic regression; standard deviations in parentheses}
  \label{tab:sim:summary_time}
  \begin{tabular}{lllll}
    \toprule
    & NICE ($\bar A= \bar 0$) & ICE ($\bar A= \bar 0$) & NICE ($\bar A = \bar 1$) & ICE ($\bar A = \bar 1$)\\
    \midrule
    Complete & 11.79 (2.27) & 4.56 (0.95) & 11.83 (2.32) & 4.26 (0.84)\\
    \addlinespace
    \multicolumn{5}{l}{(\textit{model-fitting time})} \\
    Case--control ($J=20$) & 6.61 (1.90) & 1.78 (0.57) & 6.64 (2.00) & 1.72 (0.73)\\
    Case--control ($J=10$) & 5.83 (1.74) & 1.17 (0.60) & 5.83 (2.17) & 1.14 (0.73)\\
    Case--control ($J=5$) & 5.89 (1.34) & 0.83 (0.49) & 5.90 (1.54) & 0.80 (0.54)\\
    \addlinespace
    \multicolumn{5}{l}{(\textit{total runtime})} \\ 
    Case--control ($J=20$) & 8.09 (2.66) & 2.60 (1.20) & 8.11 (2.76) & 2.55 (1.46)\\
    Case--control ($J=10$) & 6.91 (2.34) & 1.94 (1.22) & 6.94 (2.99) & 1.94 (1.47)\\
    Case--control ($J=5$) & 6.77 (1.75) & 1.54 (1.01) & 6.79 (1.95) & 1.50 (1.01)\\
    \bottomrule
  \end{tabular}
\end{table}

\begin{table}[htbp]
  \centering
  \caption{Simulation: Summary of the computation time (in mins) with superlearner; standard deviations in parentheses}
  \label{tab:sim:summary_time_sl}
  \begin{tabular}{lrr}
    \toprule
    & ICE ($\bar A = \bar 0$) & ICE ($\bar A = \bar 1$)\\
    \midrule
    Complete & 63.32 (6.19) & 62.97 (6.32)\\
    \addlinespace
    \multicolumn{3}{l}{(\textit{model-fitting time})} \\
    Case--control ($J=20$) & 5.42 (0.75) & 5.58 (0.79)\\
    Case--control ($J=10$) & 3.12 (0.51) & 3.21 (0.52)\\
    \addlinespace
    \multicolumn{3}{l}{(\textit{total runtime})} \\
    Case--control ($J=20$) & 5.44 (0.76) & 5.60 (0.80)\\
    Case--control ($J=10$) & 3.14 (0.52) & 3.23 (0.53)\\
    \bottomrule
  \end{tabular}
\end{table}

The computation time in Table~\ref{tab:sim:summary_time} and Table~\ref{tab:sim:summary_time_sl} referred to the average runtime for a single subsample run. However, because we used the bootstrap to estimate standard errors, the total computation time for one Monte Carlo run scaled linearly with the number of bootstrap replications. As expected, the overall computation time increased considerably when using superlearner compared to the logistic regression results, due to the added complexity of ensemble learning. Our subsampling approach achieved a dramatic reduction in computation time relative to the complete-data implementation. These results highlighted the practical efficiency and scalability of our method, particularly in scenarios involving complex models or large-scale datasets where full-sample estimation becomes computationally demanding.

\paragraph{Risk Estimates}

Table~\ref{tab:sim:summary_ice_nice_est} presents the estimates of the ICE and NICE estimators along with their bootstrap standard errors. The results indicate that increasing \(J\) from 5 to 20 reduces the standard errors, bringing them closer to those from the complete-data scenario. The boxplots for the bootstrap mean across the 100 simulated datasets were given in Figure~\ref{fig:sim:boxplot_risk}. For both the ICE and NICE estimators, the subsampling versions yield results comparable to their complete-data counterparts.

\begin{figure}[htbp]
  \centering
  \begin{subfigure}[t]{0.5\textwidth}
    \centering
    \includegraphics[width=\linewidth]{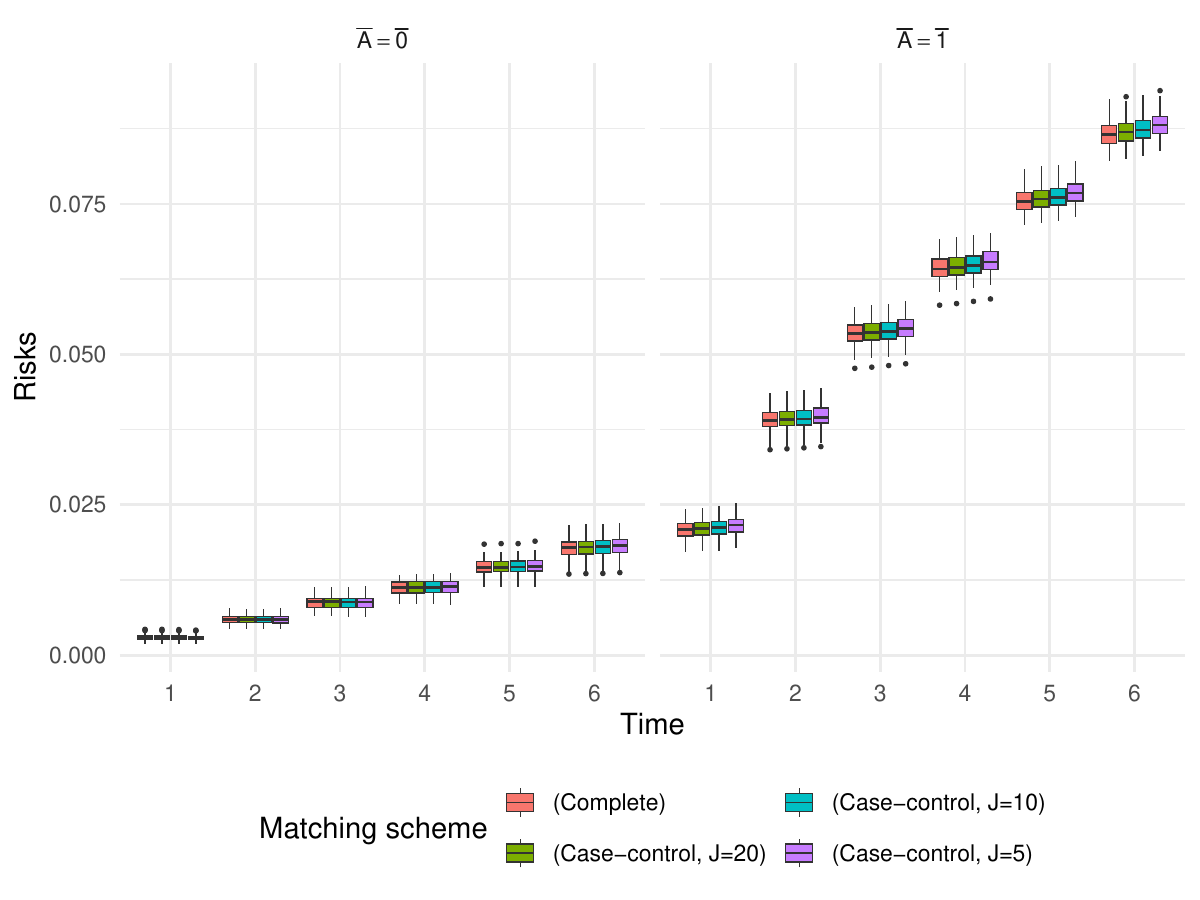}
    \caption{ICE}
  \end{subfigure}

  \begin{subfigure}[t]{0.5\textwidth}
    \centering
    \includegraphics[width=\linewidth]{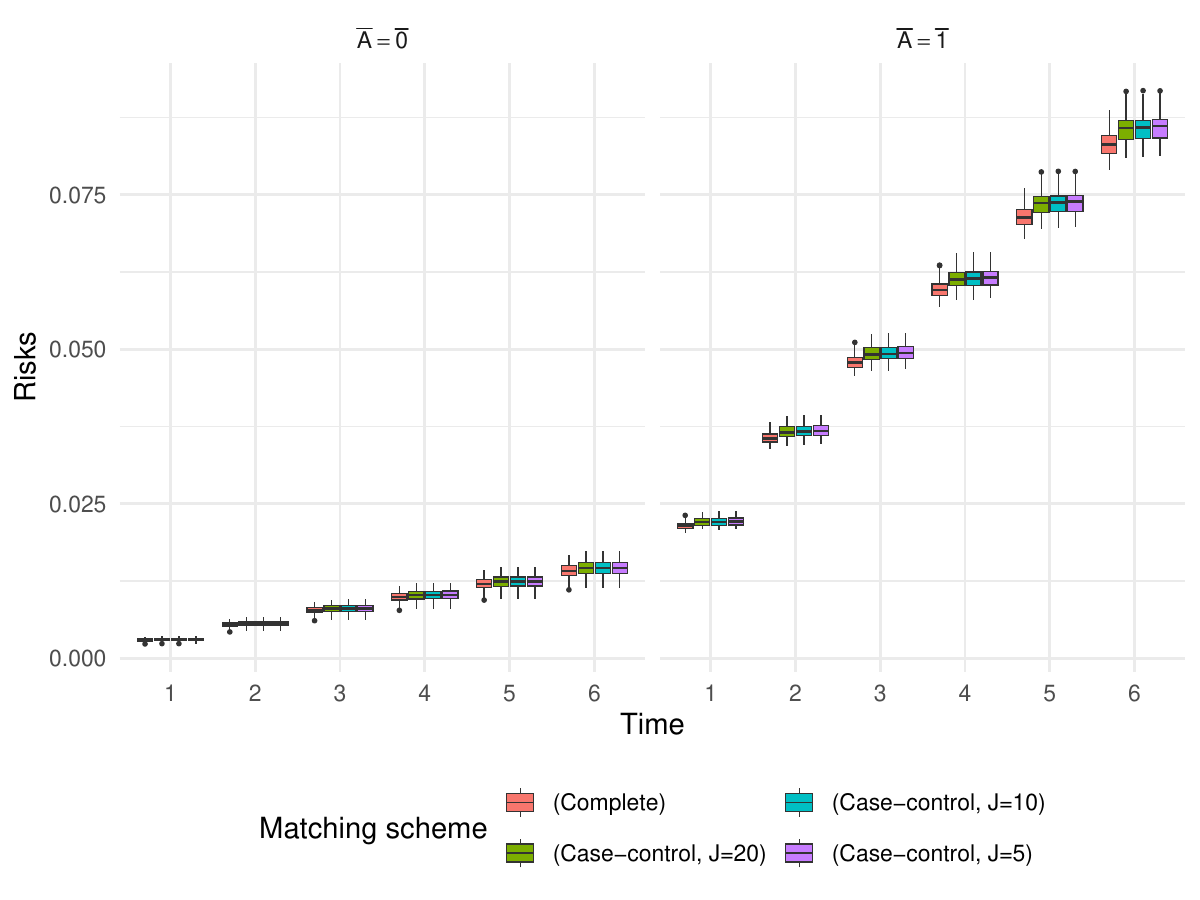}
    \caption{NICE}
  \end{subfigure}%
  \caption{Simulation: Risk estimates for always or never exposed to treatment (low-prevalence scenario).}
  \label{fig:sim:boxplot_risk}
\end{figure}

\begin{table*}[!t]
  \centering \small
  \caption{Simulation: Summary of the average of the ICE and NICE estimates based on bootstrap mean (multiplied by 100), with average of bootstrap SE in parentheses}
  \label{tab:sim:summary_ice_nice_est}
  \begin{tabular}{rllll|llll}
    \toprule
    & \multicolumn{4}{c}{\textbf{ICE}} & \multicolumn{4}{c}{\textbf{NICE}} \\
    & \multicolumn{3}{c}{Case--control}& Complete & \multicolumn{3}{c}{Case--control} & Complete \\
    Time    & ($J=5$) & ($J=10$) & ($J=20$) & & ($J=5$) & ($J=10$) & ($J=20$) & \\
    \midrule
    & \multicolumn{8}{c}{$\bar A = \bar 1$}\\
    1 & 2.15 (0.19) & 2.11 (0.17) & 2.09 (0.15) & 2.08 (0.14) & 2.22 (0.10) & 2.21 (0.09) & 2.21 (0.08) & 2.15 (0.08)\\
    2 & 3.97 (0.21) & 3.94 (0.19) & 3.93 (0.18) & 3.91 (0.17) & 3.69 (0.13) & 3.68 (0.12) & 3.67 (0.12) & 3.57 (0.11)\\
    3 & 5.44 (0.22) & 5.40 (0.21) & 5.38 (0.20) & 5.36 (0.19) & 4.96 (0.16) & 4.95 (0.15) & 4.94 (0.15) & 4.80 (0.14)\\
    4 & 6.55 (0.24) & 6.49 (0.22) & 6.46 (0.21) & 6.43 (0.20) & 6.17 (0.20) & 6.16 (0.19) & 6.15 (0.18) & 5.97 (0.17)\\
    5 & 7.69 (0.26) & 7.62 (0.24) & 7.59 (0.23) & 7.55 (0.22) & 7.39 (0.23) & 7.38 (0.22) & 7.37 (0.22) & 7.15 (0.20)\\
    6 & 8.82 (0.27) & 8.74 (0.25) & 8.70 (0.24) & 8.66 (0.23) & 8.61 (0.27) & 8.59 (0.26) & 8.58 (0.25) & 8.33 (0.23)\\
    \midrule
    & \multicolumn{8}{c}{$\bar A = \bar 0$}\\
    1 & 0.29 (0.05) & 0.29 (0.05) & 0.30 (0.05) & 0.30 (0.04) & 0.31 (0.03) & 0.31 (0.03) & 0.31 (0.03) & 0.30 (0.03)\\
    2 & 0.58 (0.08) & 0.59 (0.08) & 0.59 (0.07) & 0.59 (0.07) & 0.56 (0.05) & 0.56 (0.05) & 0.56 (0.05) & 0.55 (0.05)\\
    3 & 0.87 (0.11) & 0.87 (0.10) & 0.87 (0.10) & 0.87 (0.10) & 0.80 (0.07) & 0.80 (0.07) & 0.80 (0.07) & 0.78 (0.07)\\
    4 & 1.13 (0.14) & 1.13 (0.13) & 1.13 (0.12) & 1.13 (0.12) & 1.02 (0.09) & 1.02 (0.09) & 1.02 (0.09) & 0.99 (0.09)\\
    5 & 1.47 (0.17) & 1.46 (0.16) & 1.46 (0.15) & 1.45 (0.14) & 1.24 (0.11) & 1.24 (0.11) & 1.24 (0.11) & 1.20 (0.10)\\
    6 & 1.82 (0.19) & 1.80 (0.18) & 1.79 (0.17) & 1.78 (0.16) & 1.45 (0.13) & 1.45 (0.13) & 1.45 (0.13) & 1.41 (0.12)\\
    \bottomrule
  \end{tabular}
\end{table*}

\paragraph{Higher Prevalence Scenarios}
The results above focus on the low-prevalence setting, where the event prevalence is approximately 1\%. To assess robustness across rarity levels, we also considered two additional scenarios that differ from the low-prevalence configuration only in the intercept of the $Y_t$ model: a medium-prevalence setting with prevalence approximately 2.5\% (intercept $=-5$) and a high-prevalence setting with prevalence approximately 6\% (intercept $=-4$). 
The medium-prevalence results lead to conclusions qualitatively similar to those in the low-prevalence setting. In particular, the subsampling-based estimators remain substantially faster than the complete-data estimator, as shown in Table~\ref{tab:sim:summary_time_medium}, and the risk-estimate boxplots closely track those from the complete-data analysis, as shown in Figure~\ref{fig:sim:boxplot_risk_medium}.
The high-prevalence results, where the computational advantage of subsampling diminishes as expected, are reported in Appendix~\ref{appd:more_sim}.

\begin{table}[htbp]
  \centering \scriptsize
  \caption{Simulation: Summary of the computation time for the medium-prevalence scenario (in secs); standard deviations in parentheses}
  \vspace{2mm}
  \label{tab:sim:summary_time_medium}
  \begin{tabular}{lrrrr}
    \toprule
    & NICE, $\bar{A}=\bar{0}$ & ICE, $\bar{A}=\bar{0}$ & NICE, $\bar{A}=\bar{1}$ & ICE, $\bar{A}=\bar{1}$\\
    \midrule
    Complete & 11.26 (1.80) & 4.22 (1.03) & 11.28 (1.97) & 3.91 (1.00)\\
    \addlinespace
    \multicolumn{5}{l}{(\textit{model-fitting time})} \\
    Case--control ($J=20$) & 8.27 (2.23) & 3.02 (0.90) & 8.30 (2.32) & 2.85 (0.89)\\
    Case--control ($J=10$) & 6.92 (1.96) & 2.04 (0.86) & 6.89 (1.97) & 1.95 (1.01)\\
    Case--control ($J=5$)  & 6.16 (1.59) & 1.40 (0.47) & 6.15 (1.33) & 1.35 (0.47)\\
    \addlinespace
    \multicolumn{5}{l}{(\textit{total runtime})} \\
    Case--control ($J=20$) & 9.92 (2.78) & 3.89 (1.26) & 9.91 (2.75) & 3.74 (1.28)\\
    Case--control ($J=10$) & 8.48 (2.62) & 2.94 (1.53) & 8.50 (2.93) & 2.86 (1.75)\\
    Case--control ($J=5$)  & 7.35 (2.08) & 2.04 (0.91) & 7.33 (1.79) & 1.98 (0.86)\\
    \bottomrule
  \end{tabular}
\end{table}

\begin{figure}[htbp]
    \centering
    \begin{subfigure}[t]{0.5\textwidth}
        \centering
        \includegraphics[width=\linewidth]{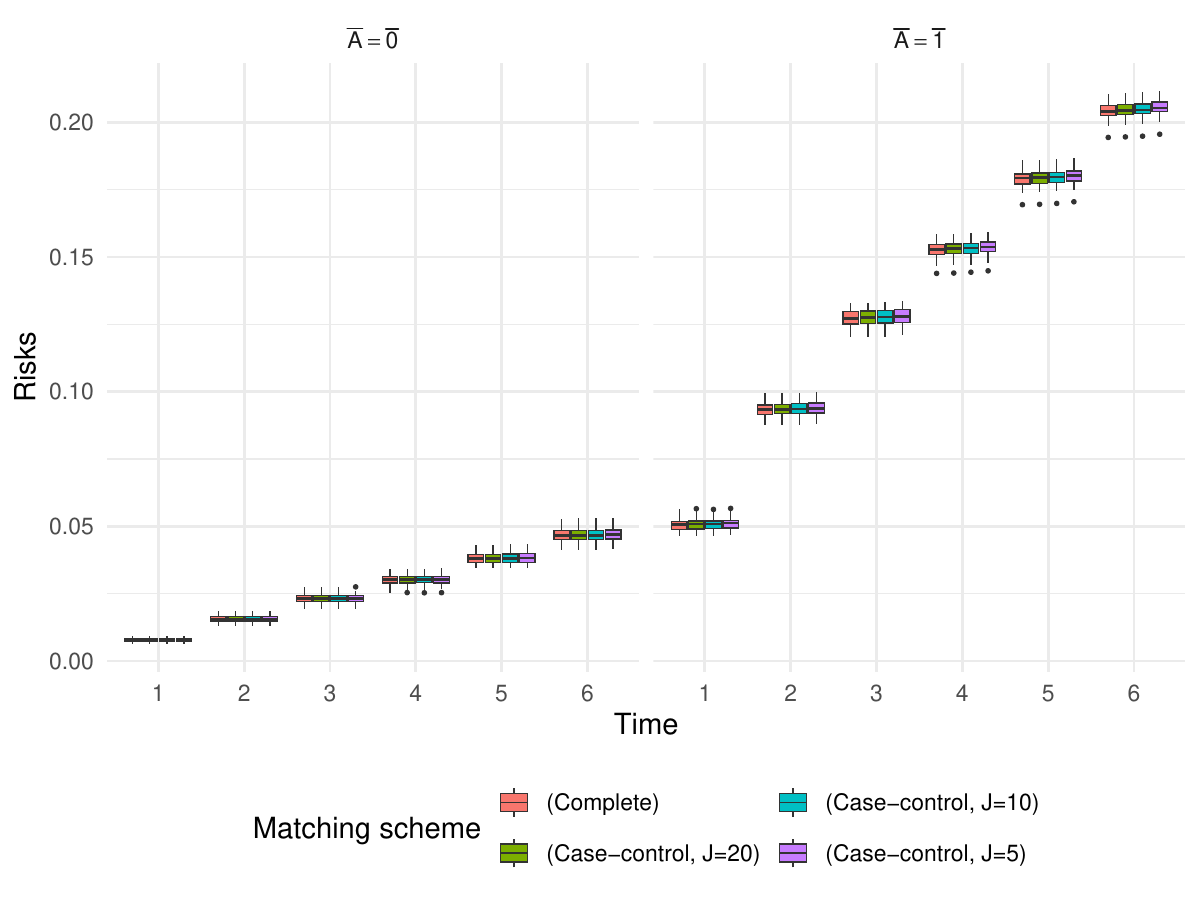}
        \caption{ICE}
    \end{subfigure}
    \begin{subfigure}[t]{0.5\textwidth}
        \centering
        \includegraphics[width=\linewidth]{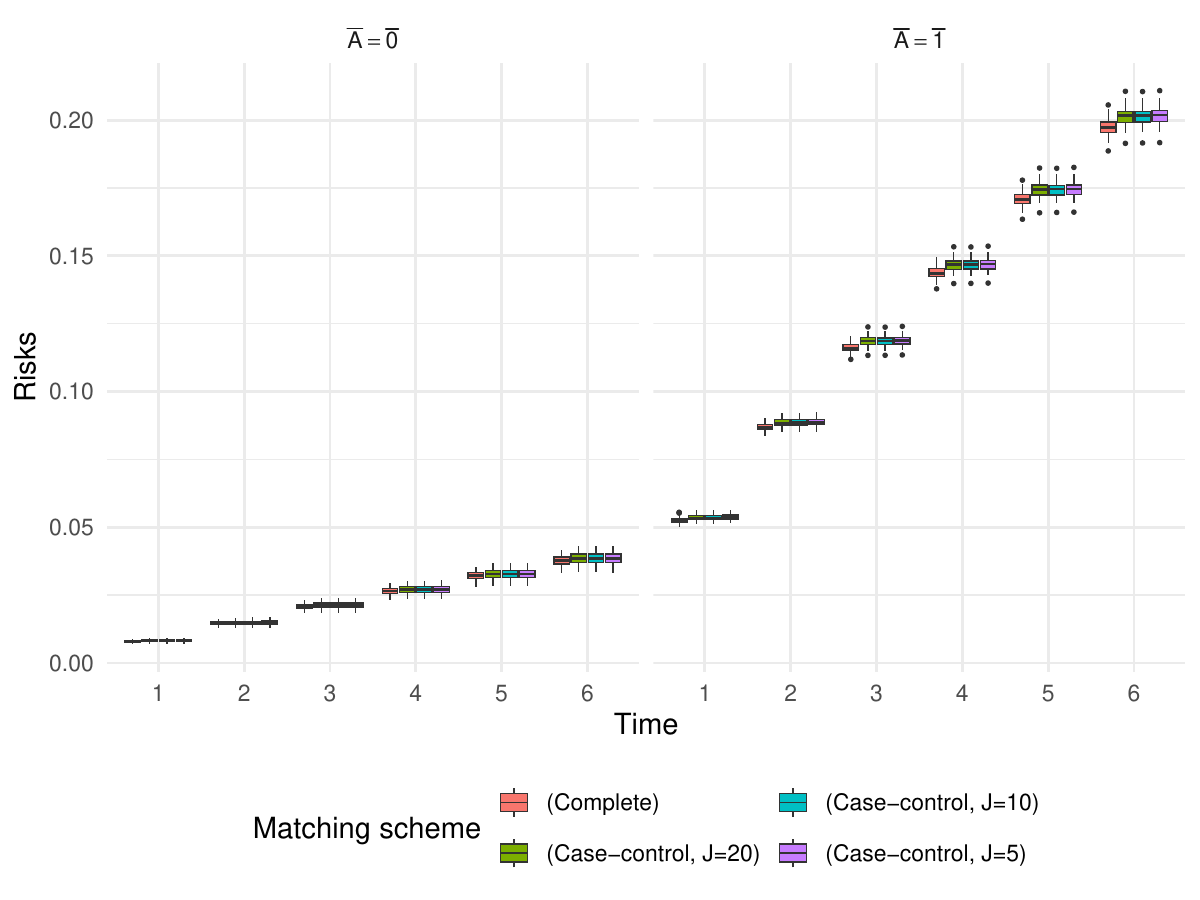}
        \caption{NICE}
    \end{subfigure}%
    \caption{Simulation: Risk estimates for always or never exposed to treatment (medium-prevalence scenario).}
    \label{fig:sim:boxplot_risk_medium}
\end{figure}

\subsection{Real-World Longitudinal Data}
We applied the proposed subsampling ICE estimators to a real-world cohort of 127,399 U.S. veterans who were discharged from short-term psychiatric hospital admissions at Veterans Health Administration (VHA) hospitals between January 1, 2016, and December 31, 2017. Patients were followed until December 31, 2018 or the event of interest, whichever occurred first. 
Death by suicide was the event of interest, and presence of any SBDH \citep{Bompelli2021} was treated as the treatment $A$.
Time-fixed baseline covariates ($L_b$) include race, gender, age groups\footnote{Although age changes over time, this is deterministic and hence we consider this as a time-fixed covariate.} and marital status; baseline summary statistics for the full cohort are provided in Table~\ref{tab:va:base_cohort_stats} in Appendix~\ref{sbdh_list}. Time-varying covariates $(L_t)$ included 17 clinical comorbidities and 7 mental health disorders. 
For model fitting, we excluded the initial 6-month window, during which no suicide deaths occurred, and analyzed the subsequent five consecutive 6-month windows. The resulting restricted analysis cohort contained 125,928 individuals with $331$ suicide cases.
Interval-specific counts are reported in Table~\ref{tab:va:summary_table_va}. Following \cite{mitra2023associations}, we extracted SBDH information from both structured data (International Classification of Diseases (ICD) codes, stop codes) and unstructured data (clinical notes, using a fine-tuned transformer-based deep-learning model \citep{vaswani2017attention}). More about SBDH, clinical, and mental health disorders are available in Appendix \ref{sbdh_list}.

Given the low incidence rate of suicide in the restricted analysis cohort, we adopted an adaptive sampling ratio across different time points. For each case at time point from $1$ to $5$, we sampled $30$, $30$, $50$, $80$, and $80$ control patients, respectively.

\begin{table}[htbp]
  \centering\small
  \caption{Application: Interval-specific summary of the restricted analysis cohort ($N=125{,}928$ veterans, $331$ suicide cases across the five analyzed 6-month intervals)}
  \vspace{2mm}
  \label{tab:va:summary_table_va}
  \begin{tabular}{ccccc}
    \toprule
    Time & Total \# & \# (\%) of Cases & \# of Controls & \# (\%) of Censored\\
    \midrule
    1 & 125,928 & $104\ (0.08\%)$ & 124,351 & $1{,}473\ (1.17\%)$\\
    2 & 124,351 & $87\ (0.07\%)$  & 122,673 & $1{,}591\ (1.28\%)$\\
    3 & 122,673 & $76\ (0.06\%)$  & 102,047 & $20{,}550\ (16.75\%)$\\
    4 & 102,047 & $44\ (0.04\%)$  & 76,537  & $25{,}466\ (24.96\%)$\\
    5 & 76,537  & $20\ (0.03\%)$  & 48,008  & $28{,}509\ (37.25\%)$\\
    \bottomrule
  \end{tabular}
\end{table}
Table \ref{tab:application_time} reports computation times, showing that the subsampling ICE estimator runs in roughly one-quarter of the time needed for the complete-data version. 
Table \ref{tab:application_risks} shows the estimated risks. The point estimates are similar for both methods, but the subsampling approach has a larger standard error at the last time point.

\begin{table}[htbp]
\centering
    \caption{Application: Computation time of ICE method with complete or case--control data (in secs); standard deviations in parentheses}
    \vspace{2mm}
    \label{tab:application_time}
    \begin{tabular}{lrr}\toprule
      & \multicolumn{1}{c}{$\bar{A}=\bar{0}$} & \multicolumn{1}{c}{$\bar{A}=\bar{1}$}\\\midrule
    Complete & $24.66\ (11.98)$ & $21.25\ (13.35)$\\
    Case--control (\textit{total runtime}) & $5.73\ (1.77)$ & $5.70\ (1.83)$\\
    Case--control (\textit{model-fitting time}) & $3.64\ (1.63)$ & $3.58\ (1.77)$\\\bottomrule
    \end{tabular}
\end{table}

\begin{table}[htbp]
\centering 
\caption{Application: Risk estimates from 100 bootstrap samples (standard deviations in parentheses); all values multiplied by 100}
\vspace{2mm}
\label{tab:application_risks}
\begin{tabular}{c rrrr}\toprule
& \multicolumn{2}{c}{$\bar{A}=\bar{0}$} & \multicolumn{2}{c}{$\bar{A}=\bar{1}$}\\\cmidrule(lr){2-3}\cmidrule(lr){4-5}
Time & \multicolumn{1}{c}{Complete} & \multicolumn{1}{c}{Case--control} & \multicolumn{1}{c}{Complete} & \multicolumn{1}{c}{Case--control}\\\midrule
1 & $0.12\ (0.02)$ & $0.12\ (0.02)$ & $0.07\ (0.01)$ & $0.07\ (0.01)$\\
2 & $0.19\ (0.02)$ & $0.20\ (0.03)$ & $0.14\ (0.01)$ & $0.15\ (0.02)$\\
3 & $0.24\ (0.03)$ & $0.25\ (0.03)$ & $0.23\ (0.02)$ & $0.24\ (0.02)$\\
4 & $0.27\ (0.03)$ & $0.28\ (0.03)$ & $0.31\ (0.02)$ & $0.33\ (0.02)$\\
5 & $0.29\ (0.03)$ & $0.31\ (0.05)$ & $0.36\ (0.02)$ & $0.40\ (0.06)$\\\bottomrule
\end{tabular}
\end{table}

\section{Discussion} \label{sec:discussion}

In this paper, we develop an outcome-dependent subsampling framework for longitudinal counterfactual risk estimation under time-varying treatments in rare-event settings. When outcomes are rare and datasets are large, longitudinal causal estimators require repeated model fitting on highly imbalanced risk sets across multiple time points, rendering full-cohort estimation computationally prohibitive. By retaining all events and strategically subsampling non-events at each time point, the proposed approach directly alleviates this computational bottleneck. 

The proposed method relies on the same identification assumptions as the underlying longitudinal g-formula estimator. Therefore, unmeasured confounding or outcome-model misspecification may still bias the resulting estimates. The subsampling step, however, is designed to reduce computational cost and, with appropriate weighting, does not change the target estimand.

Our theoretical analysis shows that, under the stated assumptions and balance
conditions on the weights, the proposed longitudinal subsampling estimator is
consistent for the target counterfactual risks.
While we do not derive closed-form variance expressions or formally characterize efficiency relative to full-cohort estimators, existing theory for rare-event subsampling suggests that much of the statistical information is concentrated in the cases, with additional controls providing diminishing marginal contributions \citep{Wang2020}. Consistent with this intuition, our numerical experiments show that the proposed subsampling strategy achieves substantial computational gains without a commensurate increase in empirical standard errors across a range of sampling ratios.

In practice, bootstrap resampling is widely used to estimate standard errors, often with \(B=1000\) or more replications. In this regime, even a modest reduction in the runtime of a single analysis can translate into substantial cumulative savings. For example, in our VHA application, the runtime of one ICE fit decreased from 24.66 seconds under the complete-data analysis to 5.73 seconds under the case--control subsampling design for $\bar A=\bar 0$. Thus, $B=1000$ bootstrap replications would take approximately 6.9 hours under the complete-data analysis, compared with approximately 1.6 hours under the case--control design. Comparable savings were observed for $\bar A=\bar 1$.
This advantage becomes even more significant in online or routinely updated analyses, where models must be refitted repeatedly as new data arrive. By substantially reducing the cost of each iteration while preserving consistency, the proposed framework offers a scalable alternative to full-cohort estimation in such settings. 

The framework is also flexible: the subsample construction and weighting scheme are not unique, allowing different designs to be tailored to specific analytic objectives. For instance, one may choose whether to include censored individuals in the subsample, or whether to sample cases and controls from future time points. As long as the resulting weights satisfy the balance equations~(\ref{eq:weights}),  the estimator remains consistent under the stated assumptions; choices such as sampling with or without replacement, including or excluding censored individuals, and varying the sampling ratio primarily affect finite-sample variance, stability, and computation. This flexibility naturally raises an important direction for future research: how to select weighting and subsampling strategies to improve efficiency. Prior work on optimal sampling strategies in logistic regression models (e.g., \citealp{fithian2014local, wang2021nonuniform}) provides useful theoretical guidance.

From a practical standpoint, the iterative fitting required in longitudinal estimators can propagate estimation error across time, particularly in later periods. This issue can be amplified under severe outcome imbalance, where logistic or related classification models may become unstable or affected by separation. Penalized methods, such as Firth's logistic regression, can be used at each step to stabilize the nuisance-model fits. More generally, the subsampling framework is compatible with flexible machine learning nuisance estimators, which may improve predictive and reduce sensitivity to simple parametric model specifications in high-dimensional or complex settings.  These considerations also motivate doubly robust extensions, such as longitudinal TMLE \citep{vanderlaan2012ltmle} or AIPTW \citep{bang2005doubly}. In principle, the same subsampling weights could be incorporated into the corresponding influence-function-based estimating equations, although a formal study of such extensions is left for future work.

\begin{acks}
Research reported in this study was in part supported by the National Institutes of Health under award numbers R01MH125027, R01AG080670, and R01DA056470, and by the VA Health Systems Research under award number I01HX003711. The content is solely the responsibility of the authors and does not necessarily represent the official views of the National Institutes of Health and the US Department of Veterans Affair.
\end{acks}


\appendix

\section{Additional Weight Combinations}\label{appd:additional_weights}
Below we give 4 sets of possible weights for the illustrative example with 2 time points without censoring:

(1) Let $w_0=w_1=1$, we have
$$
S(O)=c_0w_0S(O_{0c}) + c_1w_0S(O_{1c})+(N-c_0-c_1)w_0S(O_{1m}),
$$
and $m_0 = 0, m_1 = \frac{N-c_0-c_1}{Jc_1}$.

(2) Let $w_1=m_1=0$, we have
$$
S(O) = c_0w_0 S(O_{0c})+(N-c_0)w_0S(O_{0m}),
$$
and $w_0 = 1$, $m_0 = \frac{N-c_0}{Jc_0}$.

(3) Let $m_0=m_1=1/J$, we have
$$
S(O)=c_1w_1S(O_{1c}) + c_1S(O_{1m})+c_0w_0 S(O_{0c})+c_0S(O_{0m}),
$$
where $w_1 = \frac{c_1}{N-c_0-c_1}$, $w_0 = \frac{c_1}{N-c_0-c_1} + \frac{c_0}{N-c_0}$.

(4) Let $m_0 = \frac{N-c_0}{JN}$, $m_1 = \frac{N-c_0-c_1}{J(N-c_0)}$, then $w_0 = \frac{c_0}{N}+\frac{c_1}{N-c_0}$, $w_1=\frac{c_1}{N-c_0}$,

\section{Proof of Equation (\ref{eq:estimating_equation})} \label{appd:proof_eq}
\begin{proof}
Let
\begin{align*}
S(O) &= c_1w_1S(O_{1c})+Jc_1m_1 S(O_{1m})+c_0w_0 S(O_{0c})+ Jc_0m_0 S(O_{0m})\\
&=c_1w_1S(O_{1c}) + (N-c_0-c_1)w_1S(O_{1m})+c_0w_0 S(O_{0c})\\
& \qquad +(w_0-w_1)(N-c_0)S(O_{0m}),
\end{align*}
then
{\small
\allowdisplaybreaks 
\begin{align*}
&E\{S(O)\}\\
&=c_1 w_1\int S_0(O^*)f(O^*\mid \text{die at }t=1) dO^* \\
&\quad +Jc_1 m_1 \int S_0(O^*)f(O^*\mid \text{alive at }t=1) dO^* \\
&\quad + c_0w_0 \int S_0(O^*)f(O^*\mid \text{die at }t=0) dO^* \\
&\quad + Jc_0m_0 \int S_0(O^*)f(O^*\mid \text{alive at }t=0) dO^* \\
&= \frac{c_1w_1}{\pr(\text{die at }t=1)}\int S_0(O^*)f(O^*,\text{die at }t=1) dO^* \\
&\quad + \frac{Jc_1m_1}{\pr(\text{alive at }t=1)}\int S_0(O^*)f(O^*, \text{alive at }t=1) dO^* \\
&\quad + \frac{c_0w_0}{\pr(\text{die at }t=0)}\int S_0(O^*)f(O^*,\text{die at }t=0) dO^* \\
&\quad + \frac{Jc_0m_0}{\pr(\text{alive at }t=0)}\int S_0(O^*)f(O^*, \text{alive at }t=0) dO^* \\
&= c_1w_1\frac{N}{c_1}\int S_0(O^*)f(O^*,\text{die at }t=1) dO^* \\
&\quad + Jc_1m_1\frac{N}{N-c_0-c_1}\int S_0(O^*)f(O^*, \text{alive at }t=1) dO^* \\
&\quad + c_0w_0\frac{N}{c_0}\int S_0(O^*)f(O^*,\text{die at }t=0) dO^* \\
&\quad+ Jc_0m_0\frac{N}{N-c_0}\int S_0(O^*)f(O^*, \text{alive at }t=0) dO^* \\
&=w_1 N \int S_0(O^*)f(O^*,\text{alive at }t=0) dO^* \\
&\quad + c_0w_0\frac{N}{c_0}\int S_0(O^*)f(O^*,\text{die at }t=0) dO^*\\
&\quad + Jc_0m_0\frac{N}{N-c_0}\int S_0(O^*)f(O^*, \text{alive at }t=0) dO^*\\
& = w_0 N \int S_0(O^*)f(O^*)dO^* \\
&= w_0 N E_0^*\{S_0(O^*)\}.
\end{align*}
}
If we let $w_0 = N^{-1}$, then $E\{S(O)\}=E_0^*\{S_0(O^*)\}$. But if we only need $E\{S(O)\}=w_0NE_0^*\{S_0(O^*)\}=0$, then it is not necessary to set $w_0=N^{-1}$.
\end{proof}

\section{Proof of Theorem \ref{thm:1}} \label{appd:thm_proof}
We have already known that the ICE estimator in Algorithm \ref{alg:ice_full} is consistent given that the models are correctly specified. To show that the estimator in Algorithm \ref{alg:case_control} to be consistent, we only need to prove that at each step, the expectation of the new estimating equation based on the case--control sample is zero if and only if the expectation of the original estimating equation is zero.  
\begin{proof}
Let
\begin{align*}
S(O) &= c_{k-1}w_{k-1}S(O_{k-1,c})+ K_{k-1}c_{k-1}\ell_{k-1} S(O_{k-1,s}) \\
&\quad +J_{k-1}c_{k-1}m_{k-1} S(O_{k-1,m}) +\cdots \\
&\quad +c_0w_0 S(O_{0c})+ K_0c_0\ell_0 S(O_{0s}) + J_0c_0m_0 S(O_{0m}),
\end{align*}
then
{\small
\allowdisplaybreaks 
\begin{align*}
&E\{S(O)\}\\
&=c_{k-1} w_{k-1}\int S_0(O^*)f(O^*\mid \text{die at }t=k-1) dO^* \\
&\quad + K_{k-1}c_{k-1}\ell_{k-1}\int S_0(O^*)f(O^*\mid \text{censored at }t=k-1) dO^* \\
&\quad +J_{k-1}c_{k-1} m_{k-1} \int S_0(O^*)f(O^*\mid \text{alive at }t=k-1) dO^* + \cdots \\
&\quad + c_0w_0 \int S_0(O^*)f(O^*\mid \text{die at }t=0) dO^* \\
&\quad +K_0c_0\ell_0 \int S_0(O^*)f(O^*\mid \text{censored at }t=0) dO^* \\
&\quad +J_0c_0m_0 \int S_0(O^*)f(O^*\mid \text{alive at }t=0) dO^* \\
&= \frac{c_{k-1}w_{k-1}}{\pr(\text{die at }t=k-1)}\int S_0(O^*)f(O^*,\text{die at }t=k-1) dO^* \\
&\quad + \frac{K_{k-1}c_{k-1}\ell_{k-1}}{\pr(\text{censored at }t=k-1)}\int S_0(O^*)f(O^*,\text{censored at }t=k-1) dO^* \\
&\quad + \frac{J_{k-1}c_{k-1}m_{k-1}}{\pr(\text{alive at }t=k-1)}\int S_0(O^*)f(O^*, \text{alive at }t=k-1) dO^* + \cdots\\
&\quad + \frac{c_0w_0}{\pr(\text{die at }t=0)}\int S_0(O^*)f(O^*,\text{die at }t=0) dO^* \\
&\quad + \frac{K_0c_0\ell_0}{\pr(\text{censored at }t=0)}\int S_0(O^*)f(O^*, \text{censored at }t=0) dO^*\\
&\quad + \frac{Jc_0m_0}{\pr(\text{alive at }t=0)}\int S_0(O^*)f(O^*, \text{alive at }t=0) dO^* \\
&= c_{k-1}w_{k-1}\frac{N}{c_{k-1}}\int S_0(O^*)f(O^*,\text{die at }t=k-1) dO^* \\
&\quad + K_{k-1}c_{k-1}\ell_{k-1}\frac{N}{s_{k-1}}\int S_0(O^*)f(O^*,\text{censored at }t=k-1) dO^*\\
&\quad + J_{k-1}c_{k-1}m_{k-1}\frac{N}{N-\sum_{j=0}^{k-1}c_j - \sum_{j=0}^{k-1}s_j}\int S_0(O^*)f(O^*, \text{alive at }t=k-1) dO^* \\
&\quad + \cdots + c_0w_0\frac{N}{c_0}\int S_0(O^*)f(O^*,\text{die at }t=0) dO^* \\
&\quad + K_0c_0\ell_0\frac{N}{s_0}\int S_0(O^*)f(O^*, \text{censored at }t=0) dO^*\\
&\quad + J_0c_0m_0\frac{N}{N-c_0-s_0}\int S_0(O^*)f(O^*, \text{alive at }t=0) dO^* \\
&=w_{k-1} N \int S_0(O^*)f(O^*,\text{die at }t=k-1) dO^* \\
&\quad + w_{k-1} N \int S_0(O^*)f(O^*,\text{censored at }t=k-1) dO^* \\
&\quad + w_{k-1} N \int S_0(O^*)f(O^*,\text{alive at }t=k-1) dO^* + \cdots \\
&\quad + w_0N\int S_0(O^*)f(O^*,\text{die at }t=0) dO^* \\
&\quad + w_0N\int S_0(O^*)f(O^*,\text{censored at }t=0) dO^* \\
&\quad + Jc_0m_0\frac{N}{N-c_0}\int S_0(O^*)f(O^*, \text{alive at }t=0) dO^*\\
&= w_{k-1} N \int S_0(O^*)f(O^*,\text{alive at }t=k-2) dO^* \\
&\quad +  w_{k-2} N \int S_0(O^*)f(O^*,\text{die at }t=k-2) dO^*\\
&\quad +  w_{k-2} N \int S_0(O^*)f(O^*,\text{censored at }t=k-2) dO^* \\
&\quad + (w_{k-2}-w_{k-1})N\int S_0(O^*)f(O^*,\text{alive at }t=k-2) dO^* \\
&\quad + \cdots + w_0N\int S_0(O^*)f(O^*,\text{die at }t=0) dO^* \\
&\quad + w_0N\int S_0(O^*)f(O^*,\text{censored at }t=0) dO^* \\
&\quad + (w_{0}-w_{1})N \int S_0(O^*)f(O^*, \text{alive at }t=0) dO^*\\
& = w_0 N \int S_0(O^*)f(O^*)dO^* \\
&= w_0 N E_0^*\{S_0(O^*)\}
\end{align*}
}
Therefore, $E\{S(O)\}=0$ if and only if $E_0^*\{S_0(O^*)\}=0$.
\end{proof}

\section{More Numerical Examples} \label{appd:more_sim}

Here we report the high-prevalence scenario.
\begin{equation*}
    Y_t\sim \mathrm{Ber}\bigl(\mathrm{expit}(-4+2A_t +  {\beta}^\top  {L}_{t-1})\bigr),
\end{equation*}
yielding a prevalence rate of about $6\%$. Results are shown in Table~\ref{tab:sim:summary_time_high} and Figure~\ref{fig:sim:boxplot_risk_high}. Reference R code that reproduces the simulation results is publicly available at \url{https://github.com/XiaohuiYin1998/MatchedGFormula}.

\begin{table}[htbp]
  \centering
  \caption{Simulation: Summary of the computation time for high prevalence scenario (in secs); standard deviations in parentheses}
  \vspace{2mm}
   \scriptsize
  \label{tab:sim:summary_time_high}
  \begin{tabular}{lrrrr}
    \toprule
    & NICE, $\bar{A}=\bar{0}$ & ICE, $\bar{A}=\bar{0}$ & NICE, $\bar{A}=\bar{1}$ & ICE, $\bar{A}=\bar{1}$\\
    \midrule
    \addlinespace
    Complete & 11.06 (2.33) & 4.44 (2.44) & 11.22 (2.60) & 4.01 (2.26)\\
    \addlinespace
    \multicolumn{5}{l}{(\textit{model-fitting time})} \\
    Case--control (J=20) & 12.38 (2.35) & 5.21 (1.94) & 12.41 (2.54) & 4.79 (1.88)\\
    Case--control (J=10) & 8.30 (1.99) & 3.26 (1.63) & 8.27 (1.83) & 3.02 (1.59)\\
    Case--control (J=5) & 7.31 (1.76) & 2.44 (1.30) & 7.30 (1.70) & 2.33 (1.40)\\
    \addlinespace
    \multicolumn{5}{l}{(\textit{total runtime})} \\
    Case--control (J=20) & 14.20 (2.94) & 6.49 (2.72) & 14.25 (3.09) & 6.08 (2.64)\\
    Case--control (J=10) & 9.96 (2.60) & 4.31 (2.27) & 9.93 (2.48) & 4.06 (2.23)\\
    Case--control (J=5) & 8.94 (2.42) & 3.36 (1.98) & 8.91 (2.28) & 3.25 (2.24)\\
    \bottomrule
  \end{tabular}
  \\[2pt]
  {\footnotesize Remark: Since the prevalence is approximately 5--8\%, when $J=20$
  the sample size exceeded the original sample sizes, and consequently
  the time also exceeded that of the complete-data version.}
\end{table}

\begin{figure}[!htbp]
    \centering
    \begin{subfigure}[t]{0.5\textwidth}
        \centering
        \includegraphics[width=\linewidth]{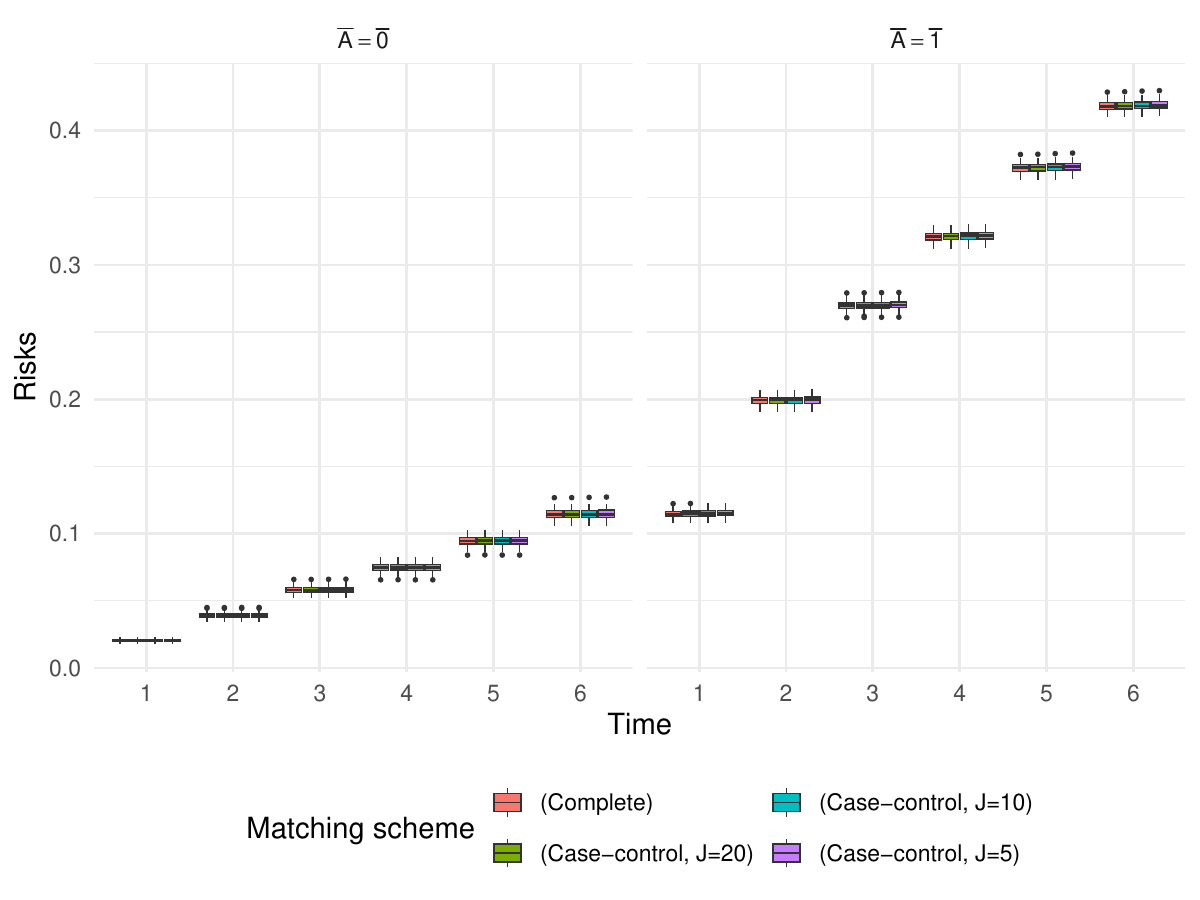}
        \caption{ICE}
    \end{subfigure}
    \begin{subfigure}[t]{0.5\textwidth}
        \centering
        \includegraphics[width=\linewidth]{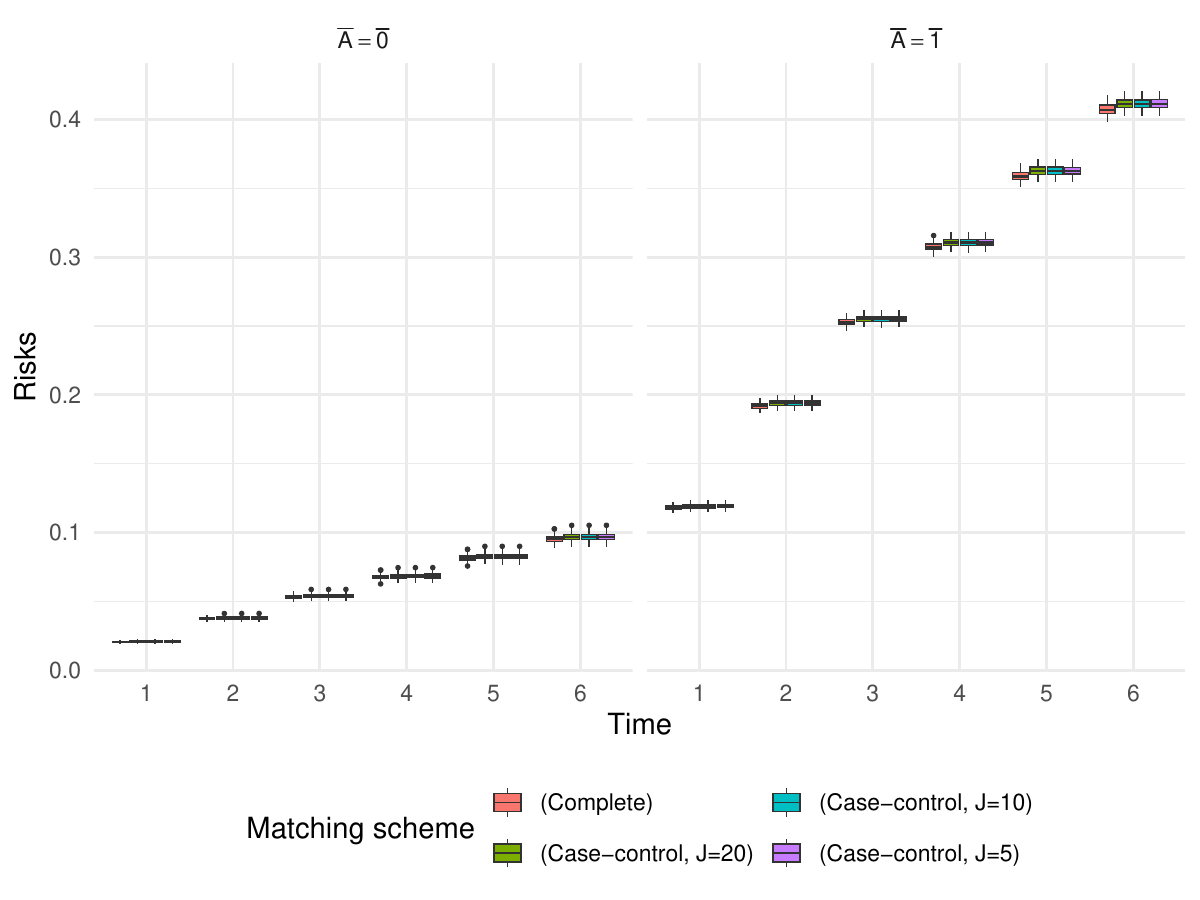}
        \caption{NICE}
    \end{subfigure}%
    \caption{Simulation: Risk estimates for always or never exposed to treatment for high prevalence scenario.}
    \label{fig:sim:boxplot_risk_high}
\end{figure}

\section{SBDH, Clinical Comorbidities, and Mental Health Disorders }
\label{sbdh_list}
\textbf{Structured SBDH categories}. We used ICD-10-CM Z codes \footnote{https://www.bcbsnm.com/docs/provider/nm/icd-10-z-codes.pdf} to extract information about 10 structured SBDH categories -
\begin{enumerate}
    \item Problems related to education and literacy (\textbf{Z55}),
    \item Problems related to employment and unemployment (\textbf{Z56}),
    \item Occupational exposure to risk factors (\textbf{Z57}),
    \item Problems related to housing and economic circumstances (\textbf{Z59}),
    \item Problems related to social environment (\textbf{Z60}),
    \item Problems related to upbringing (\textbf{Z62}),
    \item Other problems related to primary support group, including family circumstances (\textbf{Z63}),
    \item Problems related to certain psychosocial circumstances (\textbf{Z64}),
    \item Problems related to other psychosocial circumstances (\textbf{Z65}), and
    \item Problems related to medical facilities and other health care (\textbf{Z75})
\end{enumerate}

\textbf{NLP-extracted SBDH categories}. We considered 12 NLP-extracted SBDH categories - Social isolation, job or financial insecurity, housing instability, legal problems, violence, barriers to care, transition of care, food insecurity, psychiatric symptoms, substance abuse, pain, and patient disability. Definitions and sample examples of all SBDH categories as well as the deep-learning model used to extract this SBDH information are available in the appendix section of the study by Mitra et al. \citep{mitra2023associations}.

\textbf{Clinical comorbidities and mental health disorders}. From the Charlson Comorbidity Index \citep{quan2011updating}, we included 17 clinical comorbidities: acute myocardial infarction, congestive heart failure, peripheral vascular disease, cerebrovascular disease, dementia, chronic obstructive pulmonary disease, rheumatoid disease, peptic ulcer disease, mild liver disease, diabetes without complications, diabetes with complications, hemiplegia or paraplegia, kidney disease, cancer, moderate or severe liver disease, metastatic solid tumor, and AIDS/HIV. We considered 7 mental health disorders \citep{blosnich2020social}: major depressive disorder, alcohol use disorder, drug use disorder, anxiety disorder, posttraumatic stress disorder, schizophrenia, and bipolar disorder.

\begin{table}[!ht]
  \centering
  \caption{Application: Baseline summary statistics of the full VHA cohort ($N=127{,}399$ veterans, $440$ suicide cases). The restricted analysis cohort in Table~\ref{tab:va:summary_table_va} is obtained from this full cohort by excluding the initial 6-month windows with no suicide deaths and analyzing the subsequent five 6-month follow-up intervals}
  \vspace{2mm}
  \label{tab:va:base_cohort_stats}
  \begin{tabular}{llll}\toprule
    \multicolumn{2}{c}{Variables} & \makecell[cl]{Case\\(n=440), \%}  & \makecell[cl]{Control\\(n=126,959), \%} \\
    \midrule
    \multirow{3}{*}{Race}
                                  & White & 369 (83.86) & 81,443 (64.15) \\
                                  & Black & 43 (9.77) & 34,536 (27.20) \\
                                  & Others & 28 (6.36) & 10,980 (8.65)\\
    \cdashline{1-4}
    \multirow{2}{*}{Sex}
                                  & Male & 403 (91.59) & 113,690 (89.55)\\
                                  & Female & 37 (8.41) & 13,269 (10.45)\\
    \cdashline{1-4}
    \multirow{7}{*}{Age}
                                  & 18-29 & 74 (16.82) & 11,956 (9.42)\\
                                  & 30-39 & 97 (22.05) & 23,731 (18.69)\\
                                  & 40-49 & 78 (17.73)  & 18,918 (14.90)\\
                                  & 50-59 & 94 (21.36) & 34,597 (27.25)\\
                                  & 60-69 & 73 (16.59) & 29,067 (22.89)\\
                                  & 70-79 & 22 (5.00) & 6,625 (5.22)\\
                                  & 79$<$ & 2 (0.45)  & 2,065 (1.63)\\
    \cdashline{1-4}
    \multirow{3}{*}{Marital Status}
                                  & Married & 71 (16.14) & 11,722 (9.23)\\
                                  & Not married & 163 (37.05) & 25,031 (19.72) \\
                                  & Unknown & 206 (46.82) & 90,206 (71.05)\\
    \bottomrule
  \end{tabular}
\end{table}

\end{document}